\documentclass[11pt]{article}
\usepackage{amssymb}
\usepackage{amsmath}
\usepackage{amsthm}
\usepackage{mathrsfs}  
\usepackage{verbatim}
\usepackage{setspace}
\usepackage{color}   

\usepackage[dvipsnames]{xcolor}

\usepackage[authoryear]{natbib}

\usepackage{hyperref}
\usepackage{graphicx}
\usepackage{tabularx}
\usepackage{booktabs}
\usepackage[normalem]{ulem}

\usepackage[margin=1in]{geometry}
\usepackage{multicol}
\usepackage{multirow}
\usepackage{bbm}
\hypersetup{
	colorlinks=true,
	linktoc=all,     
	linkcolor=blue, 
}

\usepackage[margin=1cm]{caption}
\newcommand{\N}{{\mathbb N}}

\newcommand{\bE}{{\mathbb E}}
\newcommand{\FF}{{\mathcal F}}

\newcommand\ta{{\tilde{a}}}
\newcommand\ha{{\hat{a}}}
\newcommand\hb{{\hat{b}}}
\newcommand\tv{{\tilde{v}}}

\allowdisplaybreaks

\newtheoremstyle{Definition}{\topsep}{\topsep}{\itshape}{0pt}{\itshape}{}{}{}
\newtheorem{definition}{Definition}[section]

\newtheorem{theorem}{Theorem}[section]
\numberwithin{figure}{section}
\numberwithin{equation}{section}

\title{Uniqueness and Existence of Linear Equilibrium with a Constrained Trader}
\author{Heeyoung Kwon\footnote{Rutgers Univeristy, Department of Mathematics, Piscataway, NJ 08854 USA, Email: hk1001@math.rutgers.edu} \,\, and Jin Hyuk Choi\footnote{Ulsan National Institute of Science and Technology, Department of Mathematical Sciences, Email: jchoi@unist.ac.kr}}
\date{\today}

\begin{document}
	\maketitle
	
	\begin{abstract}
We study a discrete-time financial market with a single constrained trader, competitive market makers, and noise traders. Within the class of linear equilibria, the equilibrium structure is shown to be uniquely determined by two state variables: the market maker's expectation of the trader's remaining demand and the residual demand beyond this expectation. This discrete-time uniqueness result aligns with its continuous-time analogue, indicating that the latter may emerge as the unique limit within the same class. We also prove the existence of a linear equilibrium, providing formal support to numerical and empirical findings in related work.
\end{abstract}
	
	

\section{Introduction}

\citet{Kyle} is a cornerstone of market microstructure theory, providing a tractable dynamic model in which an informed trader strategically interacts with competitive market makers. Its continuous-time formulation in \citet{Back} further clarified the link between information flow and price formation. However, the original Kyle model does not capture certain empirically documented features,\footnote{See, for example, \citet{Jain, Hasbrouck, Barardehi}.} such as a time-varying price impact and U-shaped intraday trading patterns. Subsequent extensions have generated these features through various modifications.\footnote{See, for example, \citet{BackBaruch, Caldentey, Collin, CetinDanilova, Cetin}.} Among these, \citet{Choi_JFE} and \citet{Choi} show that introducing a trader with a random terminal trading target produces both time variation in price impact and U-shaped execution profiles, thereby aligning the model's predictions more closely with empirical evidence.

We study a discrete-time Kyle-type model in which a {\it constrained trader} faces a random terminal trading target. This setting is a simplified version of the discrete-time model in \citet{Choi_JFE}, where both an informed trader and a trader with a terminal target are present. By focusing on the latter alone, we isolate the strategic effects of the trading constraint while preserving the core price formation dynamics of the original framework. The model is also the discrete-time counterpart of \citet{Choi}, which studies a continuous-time market with the same type of constrained trader.

Our first main result is a structural uniqueness theorem: within the entire class of linear equilibria, the equilibrium can depend only on two state variables, namely the market maker's expectation of the trader's remaining demand and the residual amount the trader must execute beyond this expectation. This result rules out alternative linear structures that could generate different price dynamics or trading strategies. An important implication is that the continuous-time equilibrium in \citet{Choi} is not merely one possible specification, but the only linear structure consistent with equilibrium in the discrete-time analogue. This strengthens the credibility of the qualitative properties documented in \citet{Choi}, as they hold across the entire linear class rather than being artifacts of a particular specification.

Our second contribution is a proof of equilibrium existence. In the richer framework of \citet{Choi_JFE}, the interaction between two strategic traders made such a proof intractable. By working with a simpler but still economically meaningful setting, we close this gap, thereby placing the empirical implications in \citet{Choi_JFE} on a firmer theoretical basis. The proof is nontrivial because the equilibrium conditions form a two-dimensional boundary value problem, which would normally be difficult to solve using standard shooting methods. The key step is to identify a scaling property that reduces the effective dimensionality, making the construction feasible.

\section{Model} 

We consider a discrete-time market with $N\in\N$ trading dates. There are three types of market participants: a constrained trader, noise traders, and competitive market makers.

The constrained trader has a terminal trading target at time $N$ given by a random variable $\ta\sim N(0,\sigma_a^2)$. The asset's fundamental value is $\tv \sim N(0,\sigma_v^2)$ and is publicly revealed after the final date. The constrained trader observes $\ta$ at the beginning but cannot directly observe $\tv$. The pair $(\ta,\tv)$ is jointly normal with correlation coefficient $\rho \in (0,1]$, so the constrained trader has partial information about $\tv$ through the correlation.  
	
The noise traders submit exogenous net orders $\Delta w_n := w_n-w_{n-1}$, where $\{w_n\}_{n=1}^N$ is a discrete-time Brownian motion with independent increments. Each increment $\Delta w_n$ is normally distributed with mean zero and variance $\sigma_w^2$, independent of $(\tv, \ta)$. 

At each time $n$, the market makers observe the aggregate order flow
\begin{align}
y_n := \Delta \theta_n + \Delta w_n,\ n=1,\cdots,N,
\end{align}
where $\Delta \theta_n$ is the constrained trader's order at time $n$. The market maker's information filtration is
\begin{align*}
\mathcal{F}_n^M := \sigma\left(y_1,\cdots,y_n\right),
\end{align*}
and the constrained trader's information filtration is
\begin{align*}
\mathcal{F}_n^I := \sigma\left(\ta, y_1,\cdots,y_{n-1}\right).
\end{align*}
The market makers are competitive and risk neutral, and set the price according to
	\begin{align}\label{P_n}
		p_n = \mathbb{E}\left[\tilde{v} | \mathcal{F}_n^M\right],\quad  n=1,\cdots,N.
	\end{align}
The constrained trader starts from $\theta_0=0$ and must satisfy $\theta_N=\ta$. The expected profit is
\begin{align}
	\bE \Big[ \sum_{n=1}^{N} (\tv-p_n) \Delta \theta_n  \Big | \FF_1^I \Big]
	&= \bE \Big[\tv \,\theta_N - p_N \theta_N +  \sum_{n=1}^{N} \theta_{n-1}\Delta p_n  \Big| \FF_1^I \Big] \nonumber\\
	&=\frac{\rho \sigma_v}{\sigma_a} \ta^2 - \bE\Big[ \sum_{n=1}^{N} (\ta - \theta_{n-1}) \Delta p_n \Big| \FF_1^I \Big] \label{expected_profit},
\end{align}
where the second equality uses $\theta_N=\ta$. Hence the constrained trader's optimization problem is
\begin{align}
	\min_{\theta_n \in \FF_n^I,\,\, \theta_N=\ta} \bE\Big[ \sum_{n=1}^{N} (\ta - \theta_{n-1}) \Delta p_n \Big| \FF_1^I \Big]. \label{insider_object}
\end{align}

\begin{definition}\label{def}
An equilibrium is a pair $\{\theta_n,p_n\}_{n=1}^N$ satisfying the following conditions:
		
\noindent(i) For a given pricing rule $\{p_n\}$, the strategy $\{\theta_n\}$ solves the constrained trader's problem \eqref{insider_object}.
		
\noindent(ii) For a given strategy $\{\theta_n\}$, the price process $\{p_n\}$ satisfies the market efficiency condition \eqref{P_n}.
\end{definition}
	
Following \citet{Kyle}, we restrict attention to {\it linear equilibria} in which
\begin{align} \label{general_structure}
\begin{split}
\Delta p_n = \lambda_n y_n + h_{n-1},
\end{split}
\end{align}
where $\lambda_n$ is constant and $h_{n-1}$ is linear in $y_1,\cdots, y_{n-1}$. 	
Within this class of linear equilibria, we show in the next section that the equilibrium is structurally unique: it depends only on two state variables. The first is the market maker's expectation of the constrained trader's remaining trading demand, denoted $q_n$. The second is the residual demand beyond this expectation, given by $\ta - \theta_n - q_n$.

	\section{Analysis}
In this section, we establish the uniqueness of the equilibrium structure within the class of linear equilibria and prove the existence of the equilibrium. 

\subsection{Structural Uniqueness of Linear Equilibrium}

In Theorem~\ref{thm_unique}, we first shows that within the class of linear equilibria, the only possible equilibrium structure can be expressed in terms of the following state variables:
\begin{align}
q_n:=\bE[\ta-\theta_n | \FF_n^M] \quad \textrm{and} \quad \ta-\theta_n-q_n. \label{q_n}
\end{align} 
We interpret $q_n$ as the market maker's expectation of the constrained trader's remaining trading demand, and $\ta-\theta_n-q_n$ as the residual amount the constrained trader needs to trade beyond this expectation. We show that these two variables serves as the only sufficient statistics that characterize the equilibrium.	
	
	\begin{theorem}[Uniqueness of the structure]\label{thm_unique}
	Suppose that $\{\theta_n,p_n\}_{n=1}^N$ is a linear equilibrium of the form given in \eqref{general_structure}. Then the equilibrium must take the following form: there exist constants $\{\lambda_n, r_n, \beta_n, \alpha_n\}_{n=1}^N$ such that
\begin{align}
&\Delta p_n = \lambda_n y_n -\lambda_n\alpha_n q_{n-1},\label{P}\\
&\Delta q_n = r_n y_n -(1+r_n)\alpha_n q_{n-1},\label{Q}\\
&\Delta \theta_n = \beta_n \left(\tilde{a}-\theta_{n-1}-q_{n-1}\right) + \alpha_n q_{n-1}, \label{theta_n}
\end{align}
where $q_n:=\bE[\ta-\theta_n | \FF_n^M]$.
\end{theorem}
\begin{proof} 
See Appendix.
\end{proof}

\citet{Choi} studies the continuous-time counterpart of our model and constructs a specific linear equilibrium of the following form:
\begin{align*}
dp_t &= \lambda(t) dy_t -\lambda(t) \alpha(t) q_t dt,\\
dq_t &= r(t) dy_t -(1+r(t))\alpha(t) q_t dt,\\
d\theta_t &= \big( \beta(t) (\ta - \theta_t - q_t) + \alpha(t) q_t\big) dt,
\end{align*} 
for deterministic functions $\lambda(t),r(t),\beta(t)$, and $\alpha(t)$. This specification is the natural continuous-time analogue of the discrete-time structure characterized in Theorem~\ref{thm_unique}. In \citet{Choi}, the emergence of the two state variables $q_t$ and $\ta - \theta_t - q_t$ is not derived from a general uniqueness result, but rather taken as part of the specific solution. In contrast, Theorem~\ref{thm_unique} shows that in the discrete-time setting, these variables necessarily constitute the unique linear equilibrium structure.
This suggests that the continuous-time equilibrium in \citet{Choi} may be the unique limit of the discrete-time linear equilibria, at least within the same class.
Figure~\ref{fig1} provides numerical evidence supporting this connection. In this sense, our uniqueness result provides a theoretical justification for the specific equilibrium structure employed in \citet{Choi}, reinforcing that the qualitative insights drawn from that equilibrium, such as the trader's behavior or price impact, are not artifacts of a particular specification but hold robustly across the entire linear class.
\begin{figure}[htbp]
	\centering
	\begin{minipage}[b]{0.45\textwidth}
		\centering
		\includegraphics[width=\linewidth]{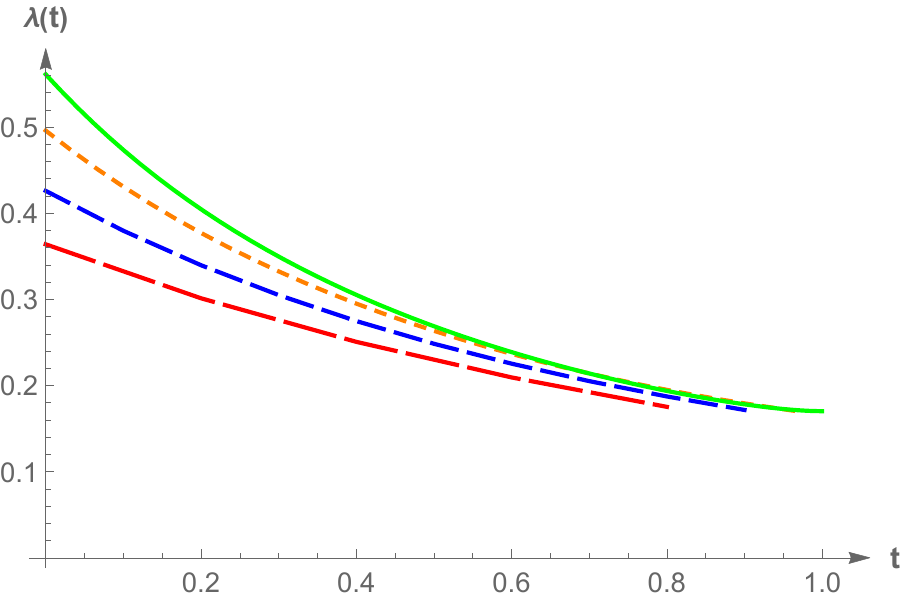}
	\end{minipage}
	\hspace{0.03\textwidth}
	\begin{minipage}[b]{0.45\textwidth}
		\centering
		\includegraphics[width=\linewidth]{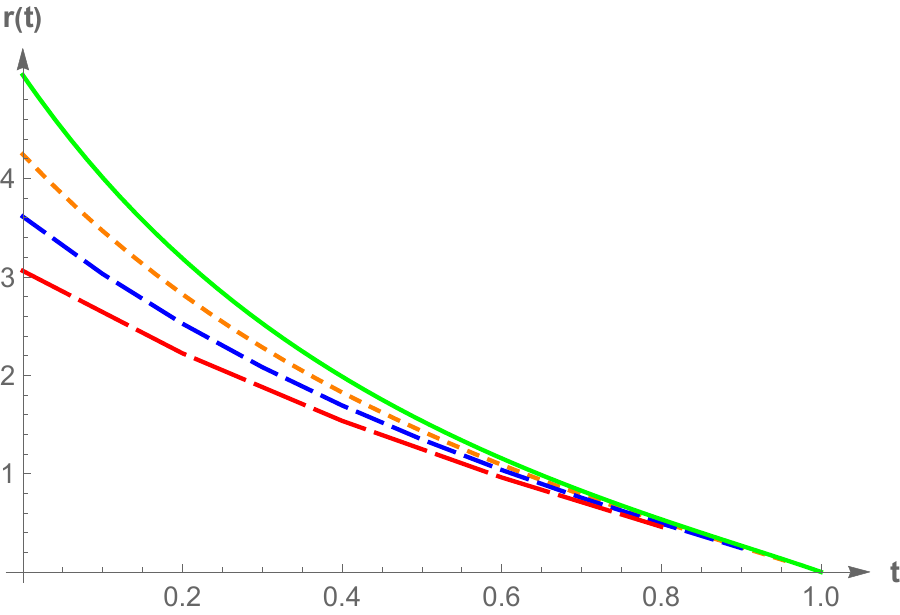}
	\end{minipage}
	\captionsetup{width=\linewidth}
	\caption{
	Comparison of continuous-time equilibrium quantities (left: $\lambda(t)$, right: $r(t)$ shown in green) with the corresponding discrete-time values $\{\lambda_n\}_n$ and $\{r_n\}_n$.
Exogenous parameters are $\sigma_w=1/\sqrt{N}$, $\sigma_v=1$, $\sigma_a=3$, $\rho=\tfrac{1}{3}$ and $N=5$ (red), $10$ (blue), $30$ (orange).	
	}
	\label{fig1}
\end{figure}

\subsection{Existence of Linear Equilibrium}

Having established the structural uniqueness, we now turn to the existence of a linear equilibrium.
We begin by deriving the equations implied by the equilibrium conditions \eqref{P_n} and \eqref{q_n}.
Define the innovation process $z_n:=y_n - \bE[y_n|F_{n-1}^M]$. From \eqref{theta_n}, we obtain:
\begin{align}
z_n &= y_n - \bE\left[ y_n \middle| \mathcal{F}_{n-1}^M \right] 
= \Delta \theta_n + \Delta w_n - \mathbb{E}\left[ \Delta \theta_n + \Delta w_n \middle| \mathcal{F}_{n-1}^M \right] \nonumber\\
&= \beta_n (\ta-\theta_{n-1}-q_{n-1}) +\alpha_n q_{n-1} + \Delta w_n - \bE\left[\beta_n (\ta-\theta_{n-1}-q_{n-1}) +\alpha_n q_{n-1}+ \Delta w_n \middle| \mathcal{F}_{n-1}^M \right] \nonumber\\
&= \beta_n \left( \tilde{a}-\theta_{n-1} -q_{n-1}\right) + \Delta w_n, \label{z}
\end{align}
where the final equality follows from $q_{n-1}=\bE[\ta-\theta_{n-1} | \FF_{n-1}^M]$. 
Due to the linear structure, all involved random variables are jointly Gaussian. It follows that $z_1,\dots,z_N$ are mutually independent, and both $\tv-p_n$ and $\ta-\theta_n-q_n$ are independent of $\sigma(z_1,\dots,z_n)$. We define 
\begin{align}
\Sigma_n^{(1)}:=\mathbb{E}\big[\left(\ta-\theta_n-q_n\right)^2\big], \quad
\Sigma_n^{(2)}:=\mathbb{E}\left[\left(\ta-\theta_n-q_n\right)\left(\tv-p_n\right)\right]. \label{Sigma_def}
\end{align}
Then the filtering identities $p_n = \mathbb{E}[\tilde{v} | \mathcal{F}_n^M]$ and $q_n = \mathbb{E}[\ta - \theta_n | \mathcal{F}_n^M]$ imply 
\begin{align}
\Delta p_n&=\bE[\tv - p_{n-1}|\FF_{n}^M]=\bE[\tv - p_{n-1} |\sigma(z_1,\dots,z_n)]=\bE[\tv - p_{n-1} |\sigma(z_n)]   \nonumber\\
& = \frac{\bE[(\tv-p_{n-1}) z_n]}{\bE[z_n^2]} z_n
= \frac{\beta_n \Sigma_{n-1}^{(2)}}{\beta_n^2 \Sigma_{n-1}^{(1)} + \sigma_w^2}z_n, \label{p_filter}\\
\Delta q_n &= \bE[ \ta-\theta_n - q_{n-1}| \FF_n^M ]=\bE[ \ta-\theta_{n-1} - q_{n-1}- \beta_n (\ta - \theta_{n-1}-q_{n-1}) - \alpha_n q_{n-1} | \FF_n^M ]\nonumber\\
&= (1-\beta_n)\bE[ \ta-\theta_{n-1} - q_{n-1}| \sigma(z_n) ] - \alpha_n q_{n-1} \nonumber\\
&= (1-\beta_n)\frac{\bE[(\ta-\theta_{n-1}-q_{n-1})z_n]}{\bE[z_n^2]} z_n - \alpha_n q_{n-1}
=\frac{(1-\beta_n)\beta_n \Sigma_{n-1}^{(1)}}{\beta_n^2 \Sigma_{n-1}^{(1)}+\sigma_w^2} z_n - \alpha_n q_{n-1}. \label{q_filter}
\end{align}
Comparing \eqref{P}-\eqref{Q} with \eqref{p_filter}-\eqref{q_filter}, we identify
\begin{align}
\lambda_n = \frac{\beta_n \Sigma_{n-1}^{(2)}}{\beta_n^2 \Sigma_{n-1}^{(1)} + \sigma_w^2}, \quad 
r_n =  \frac{(1-\beta_n)\beta_n \Sigma_{n-1}^{(1)}}{\beta_n^2 \Sigma_{n-1}^{(1)}+\sigma_w^2}. \label{lambda-r}
\end{align}
Substituting \eqref{P}, \eqref{Q}, \eqref{theta_n}, and \eqref{lambda-r} into \eqref{Sigma_def}, we derive the recursions:
\begin{align}
\Sigma_n^{(1)}=\frac{(1-\beta_n)^2 \sigma_w^2 \Sigma_{n-1}^{(1)}}{\beta_n^2 \Sigma_{n-1}^{(1)}+\sigma_w^2}, \quad 
\Sigma_n^{(2)}=\frac{(1-\beta_n) \sigma_w^2 \Sigma_{n-1}^{(2)}}{\beta_n^2 \Sigma_{n-1}^{(1)}+\sigma_w^2}. \label{Sigma_eq}
\end{align}
For notational convenience, define
\begin{align}
\xi_n:= \frac{\beta_n}{1-\beta_n} \quad \textrm{for}\quad 1\leq n \leq N-1. \label{xi_def}
\end{align}
Then from \eqref{lambda-r}, \eqref{Sigma_eq}, and \eqref{xi_def}, we obtain for $1\leq n \leq N-1$,
\begin{align}
\lambda_n =  \frac{\Sigma_n^{(2)}}{\sigma_w^2} \xi_n,  \quad
r_n= \frac{\Sigma_n^{(1)}}{\sigma_w^2} \xi_n,  \quad
\Sigma_{n-1}^{(1)}= \frac{\sigma_w^2 \Sigma_n^{(1)}(1+\xi_n)^2}{\sigma_w^2 -\Sigma_n^{(1)}\xi_n^2}, \quad
\Sigma_{n-1}^{(2)}= \frac{\sigma_w^2 \Sigma_n^{(2)}(1+\xi_n)}{\sigma_w^2 -\Sigma_n^{(1)}\xi_n^2}. \label{lr}
\end{align}

We next turn to the value function and the optimality condition. As shown in the proof of Theorem~\ref{thm_unique} in the Appendix, the constrained trader's value function has the  quadratic form in \eqref{value_form}, and the optimality condition yields \eqref{beta}, which can be rewritten in terms of $\xi_n$ as
\begin{align}
\xi_n = \frac{2(1+r_n)I_n - r_n J_n - \lambda_n}{2r_n (1+r_n)I_n - r_n(1+2r_n) J_n + \lambda_n}, \label{xi_eq}
\end{align}
subject to the second-order condition:  
\begin{align}
(1+r_n)((1+r_n)I_n - r_n J_n)>0. \label{SOC}
\end{align}
 
Substituting \eqref{lr} and \eqref{xi_def} into  \eqref{I} and \eqref{J}, the recursion for $I_n$ and $J_n$ becomes
\begin{align}
\begin{split}\label{IJ}
I_{n-1}&=\frac{(\sigma_w^2 - \Sigma_n^{(1)} \xi_n^2)(\sigma_w^2 + 2\sigma_w^2 \xi_n + \Sigma_n^{(1)} \xi_n^2)}{\sigma_w^4 (1+\xi_n)^2}I_n + \frac{\Sigma_n^{(1)} (\sigma_w^2 + \Sigma_n^{(1)} \xi_n) \xi_n^3}{\sigma_w^4 (1+\xi_n)^2}J_n, \quad
J_{n-1}=\frac{\lambda_n}{1+r_n}.
\end{split} 
\end{align}

The terminal condition $\theta_N=\ta$ implies
\begin{align}
&\beta_N=\alpha_N=1, \quad r_N=0, \quad s_N=-1, \nonumber\\
&I_{N-1}=J_{N-1}=\lambda_N=\frac{\Sigma_{N-1}^{(2)}}{\Sigma_{N-1}^{(1)}+\sigma_w^2}, \label{terminal}
\end{align}
where \eqref{terminal} follows from \eqref{IJ_N} and \eqref{lambda-r}. 
The definition \eqref{Sigma_def} gives the initial conditions:
\begin{align}
\Sigma_0^{(1)}=\bE[\ta^2]=\sigma_a^2, \quad \Sigma_0^{(2)}=\bE[\ta \tv] = \rho \, \sigma_a \sigma_v. \label{Sigma_0}
\end{align}

To summarize, proving the existence of a linear equilibrium reduces to verify that there exist sequences of constants $\{\xi_n,r_n,I_n,J_n\}_{n=1}^{N-1}$, $\{\lambda_n\}_{n=1}^N$, $\{\Sigma_n^{(1)},\Sigma_n^{(2)}\}_{n=0}^{N-1}$ such that:

\medskip

i) For $1\leq n \leq N-1$, the equations and inequality \eqref{lr}, \eqref{xi_eq}, \eqref{SOC}, and  \eqref{IJ} hold. 

\smallskip

ii) The terminal \eqref{terminal} and \eqref{Sigma_0} are satisfied. 

\medskip

Given terminal values $(\Sigma_{N-1}^{(1)},\Sigma_{N-1}^{(2)})$, the recursive equations determine all coefficients backward in time, leading to initial values $(\Sigma_0^{(1)},\Sigma_0^{(2)})$.
The problem therefore reduces to finding $(\Sigma_{N-1}^{(1)},\Sigma_{N-1}^{(2)})$ such that \eqref{Sigma_0} is satisfied.
This is a two dimensional boundary value problem, which would normally be difficult to solve via standard shooting methods.
The key observation is a scaling property that effectively reduces the problem to one dimension, allowing the construction to proceed.

The next theorem establishes that such constants indeed exist, and therefore a linear equilibrium exists. Once these constants are obtained, the remaining equilibrium parameters $\{\beta_n, \alpha_n\}$ can be computed using \eqref{xi_def} and \eqref{theta_n2}.

\medskip

\begin{theorem}\label{thm_existence}
There exists a linear equilibrium.
\end{theorem}
		
\begin{proof} 
See Appendix.
\end{proof}
	
\medskip
	
The existence result established in Theorem~\ref{thm_existence} fills an important gap in the literature. In the more complex setting of \citet{Choi_JFE}, the interaction between two strategic traders prevented a proof of existence, leaving open the possibility that the reported numerical equilibrium might not be supported by an actual solution. Our result demonstrates that, at least in the simpler but economically meaningful single-agent setting, a linear equilibrium does exist.

This finding provides theoretical support to the empirical and numerical insights developed in the earlier literature. It also highlights how structural simplification can enable formal analysis, and suggests that similar techniques may be useful in analyzing other constrained trading environments.

\section{Conclusion}

This paper provides two main contributions to the study of linear equilibria in markets with constrained traders. First, we establish that the equilibrium structure is uniquely determined within the linear class by two economically meaningful state variables. This result not only clarifies the foundation of the continuous-time equilibrium in \citet{Choi}, but also reinforces the robustness of its qualitative predictions. Second, we prove that a linear equilibrium exists in the discrete-time setting, filling a gap left open in the richer two trader model of \citet{Choi_JFE}. This existence result offers a theoretical basis for prior empirical and numerical findings, and it illustrates how carefully chosen simplifications can make otherwise intractable models analytically solvable. Future research could extend these methods to multi-agent environments or settings with additional market frictions.

\bigskip\bigskip

{\bf Conflict of interest statement}: We declare that we have no known competing financial interests or personal relationships that could have appeared to influence the work reported in this paper.\\

{\bf Data availability statement}: Data sharing not applicable to this article as no datasets were generated or analysed during the current study.\\

{\bf Funding}: This work was supported by the National Research Foundation of Korea (NRF) grant funded by the Korea government (MSIT) (No. RS-2023-00208622).

\appendix

\section{Proofs}\label{appendix_A}
	
\begin{proof}[{\bf Proof of Theorem~\ref{thm_unique}}]
The condition \eqref{P_n} and the linear pricing rule \eqref{general_structure} imply
\begin{align*}
0 &= \mathbb{E}\left[ \Delta p_n \middle| \mathcal{F}_{n-1}^M \right] = \mathbb{E}\left[ \lambda_n y_n + h_{n-1} \middle| \mathcal{F}_{n-1}^M \right] = \lambda_n\, \mathbb{E}\left[\Delta \theta_n \middle| \mathcal{F}_{n-1}^M \right] + h_{n-1},
\end{align*}
so we obtain
\begin{align}\label{h_dynamics}
h_{n-1} = -\lambda_n \, \mathbb{E}[\Delta \theta_n | \mathcal{F}_{n-1}^M].
\end{align}
The terminal constraint $\theta_N=\ta$ implies 
\begin{align}
\Delta \theta_N = \ta - \theta_{N-1}, \quad \textrm{and}\quad q_N= 0.\label{Delta_theta_N}
\end{align}
We show by backward induction that $q_n$ is a linear function of $y_1,\dots,y_n$ for all $n=1,\dots,N$. The case of $n=N$ follows from \eqref{Delta_theta_N}. Suppose that $q_n$ is a linear function of $y_1,\dots,y_n$. Then there exist a constant $r_n$ and a linear function $\tilde h_{n-1}$ of $y_1,\dots,y_{n-1}$ such that
\begin{align}
q_n=r_n y_n + \tilde h_{n-1}. \label{tilde_h}
\end{align} 
We verify that $q_{n-1}$ is a linear function of $y_1,\dots,y_{n-1}$ as follows:
\begin{align}
q_{n-1} &= \bE[\tilde{a}-\theta_{n-1} | \FF_{n-1}^M ] 
=\bE[\tilde{a}-\theta_n | \FF_{n-1}^M ] +\bE[\Delta \theta_n | \FF_{n-1}^M ]
=\bE[q_n | \FF_{n-1}^M ] +\bE[\Delta \theta_n | \FF_{n-1}^M ]\nonumber\\
&=\bE[r_n(\Delta \theta_n + \Delta w_n) + \tilde h_{n-1} | \FF_{n-1}^M ] +\bE[\Delta \theta_n | \FF_{n-1}^M ] = (1+r_n)\bE[\Delta \theta_n | \FF_{n-1}^M ] + \tilde h_{n-1} \nonumber \\
&=-\tfrac{1+r_n}{\lambda_n} h_{n-1} + \tilde{h}_{n-1}, \label{q_n-1}
\end{align}
where the final equality follows from \eqref{h_dynamics}.
Therefore, by induction, $q_n$ is a linear function of $y_1,\dots,y_n$ for all $n=1,\dots,N$. Moreover, from \eqref{tilde_h} and \eqref{q_n-1}, we obtain
\begin{align}
\Delta q_n = r_n y_n + \tfrac{1+r_n}{\lambda_n} h_{n-1}. \label{Delta_q_n}
\end{align}

We now analyze the constrained trader's value function and derive the optimality condition. Specifically, we show by backward induction that the value function in \eqref{insider_object} takes the following form: for $n=0,1,\dots,N-1$, there exist constants $\{I_n,J_n,K_n\}$ such that
\begin{align}
	 \bE\Big[ \sum_{k=n+1}^{N} (\ta - \theta_{k-1}) \Delta p_k \Big| \FF_{n+1}^I \Big] =I_n(\tilde{a}-\theta_{n}-q_{n})^2 + J_n (\tilde{a}-\theta_{n}-q_{n})q_{n} + K_n.  \label{value_form}
\end{align}

The case of $n=N-1$ follows by direct calculation. 
From \eqref{general_structure}, \eqref{h_dynamics}, and \eqref{Delta_theta_N}, we obtain
\begin{align}
\Delta p_N&=\lambda_N y_N - \lambda_N q_{N-1} =\lambda_N (\ta-\theta_{N-1}+\Delta w_n) - \lambda_N q_{N-1}. \label{Delta_p_N}
\end{align}
Then
\begin{align} 
\mathbb{E}[(\tilde{a}-\theta_{N-1}) \Delta p_N | \mathcal{F}_{N}^I]&=\lambda_N (\tilde{a}-\theta_{N-1}) ^2 - \lambda_N (\tilde{a}-\theta_{N-1})q_{N-1} \nonumber\\ 
&=\lambda_N (\tilde{a}-\theta_{N-1}- q_{N-1}) ^2 + \lambda_N (\tilde{a}-\theta_{N-1}-q_{N-1})q_{N-1}. \label{exp_N-1}
\end{align}
Thus, \eqref{value_form} holds for $n=N-1$ with 
\begin{align}
I_{N-1}=J_{N-1}=\lambda_N, \quad K_{N-1}=0. \label{IJ_N}
\end{align}

To proceed by induction, suppose that \eqref{value_form} holds. Then,
\begin{align}
&\bE\Big[ \sum_{k=n}^{N} (\ta - \theta_{k-1}) \Delta p_k \Big| \FF_{n}^I \Big] \nonumber\\
&= \bE\Big[  (\ta - \theta_{n-1}) \Delta p_n + I_n(\tilde{a}-\theta_{n}-q_{n})^2 + J_n (\tilde{a}-\theta_{n}-q_{n})q_{n} + K_n\Big| \FF_{n}^I \Big] \nonumber\\
	 &= \bE\Big[  (\ta - \theta_{n-1}) \Delta p_n + I_n(\tilde{a}-\theta_{n-1}-q_{n-1} - \Delta \theta_n  - \Delta q_n )^2 \nonumber \\
	 &\qquad\qquad + J_n (\tilde{a}-\theta_{n-1}-q_{n-1} - \Delta \theta_n  - \Delta q_n ) (q_{n-1}+ \Delta q_n) + K_n\Big| \FF_{n}^I \Big]. \label{value_n}
\end{align}
Applying \eqref{general_structure} and \eqref{Delta_q_n} to \eqref{value_n}, we obtain an expression that is quadratic in $\Delta \theta_n$. The optimality of $\Delta \theta_n$ in equilibrium implies the following first-order condition:
\begin{align}
\Delta \theta_n &= \beta_n (\ta - \theta_{n-1}-q_{n-1}) + \eta_n q_{n-1} + \gamma_n h_{n-1}\label{theta_n1}\\
\textrm{where}\quad \beta_n &= \tfrac{2(1+r_n)I_n - r_n J_n - \lambda_n}{2(1+r_n)((1+r_n)I_n-r_n J_n)}, \label{beta}\\
\eta_n &= \tfrac{(1+r_n)J_n - \lambda_n}{2(1+r_n)((1+r_n)I_n-r_n J_n)}, \quad \gamma_n = \tfrac{(1+2r_n)J_n - 2(1+r_n)I_n}{2 \lambda_n ((1+r_n)I_n-r_n J_n)}.\label{gamma}
\end{align}
Combining \eqref{h_dynamics}, \eqref{theta_n1}, and the definition $q_{n-1}=\bE[\ta-\theta_{n-1}|\FF_{n-1}^M]$, we obtain
\begin{align}
h_{n-1} = -\lambda_n \, \mathbb{E}[\Delta \theta_n | \mathcal{F}_{n-1}^M]=-\lambda_n( \eta_n q_{n-1} + \gamma_n h_{n-1}) \quad \Longrightarrow \quad h_{n-1} = -\tfrac{\eta_n \lambda_n }{1+\gamma_n \lambda_n} q_{n-1}. \label{h_q}
\end{align}
Substituting \eqref{h_q} and \eqref{gamma} into \eqref{theta_n1}, we obtain that for $1\leq n \leq N-1$,
\begin{align}
\Delta \theta_n &= \beta_n (\ta - \theta_{n-1}-q_{n-1}) +  \alpha_n q_{n-1}, \quad \textrm{where} \quad \alpha_n= 1- \tfrac{\lambda_n}{(1+r_n)J_n}.
\label{theta_n2}
\end{align} 
Plugging \eqref{theta_n2} into \eqref{h_dynamics} yields 
\begin{align*}
h_{n-1}=-\lambda_n \alpha_n q_{n-1}.
\end{align*}
This shows that \eqref{h_dynamics} and \eqref{Delta_q_n} can be rewritten as \eqref{P} and \eqref{Q}.
Finally, substituting \eqref{P}, \eqref{Q}, \eqref{theta_n2}, and \eqref{beta} into \eqref{value_n} leads to the recursive structure:
\begin{align}
	 \bE\Big[ \sum_{k=n}^{N} (\ta - \theta_{k-1}) \Delta p_k \Big| \FF_{n}^I \Big] &= I_{n-1}(\tilde{a}-\theta_{n-1}-q_{n-1})^2 + J_{n-1} (\tilde{a}-\theta_{n-1}-q_{n-1})q_{n-1} + K_{n-1}\nonumber\\
	 \textrm{where} \quad I_{n-1}&= (1-(1+r_n)^2\beta_n^2) I_n +\beta_n^2r_n(1+r_n)J_n,\label{I}\\
		J_{n-1} &= \tfrac{\lambda_n}{1+r_n},\label{J}\\
		K_{n-1} &= K_n + (I_n-J_n) r_n^2 \sigma_w^2.\label{K}
\end{align}
By induction, this completes the verification of \eqref{value_form} for all $n=0,\dots,N-1$, and confirms that the structure given in \eqref{P}, \eqref{Q}, and \eqref{theta_n} holds.
\end{proof}

\bigskip

\begin{proof}[{\bf Proof of Theorem~\ref{thm_existence}}]
We aim to find a solution using the shooting method by adjusting the value of $(\Sigma_{N-1}^{(1)},\Sigma_{N-1}^{(2)})$ to satisfy the boundary condition \eqref{Sigma_0}.

We observe that equations \eqref{lr} and \eqref{xi_eq} yield the following equation for $\xi_n$ for $1\leq n< N$:
\begin{align}
f_n(\xi_n):=2\Big(I_n+ \tfrac{\Sigma_n^{(1)}}{\sigma_w^2}(I_n-J_n) \xi_n \Big) \Big( \tfrac{\Sigma_n^{(1)}}{\sigma_w^2} \xi_n^2 - 1\Big) +\xi_n(1+\xi_n)\Big(\tfrac{\Sigma_n^{(2)}}{\sigma_w^2}-\tfrac{\Sigma_n^{(1)}}{\sigma_w^2}J_n\Big)  = 0. \label{fn_def}
\end{align}
In addition, \eqref{lr} and \eqref{IJ} together imply
\begin{align}
I_{n-1}-J_{n-1}&=\tfrac{1+2\xi_n - r_n(2+r_n)\xi_n^2}{(1+\xi_n)^2} I_n + \tfrac{r_n(1+r_n)\xi_n^2}{(1+\xi_n)^2} J_n - \tfrac{\lambda_n}{1+r_n}\nonumber\\
&=-\tfrac{(\lambda_n - r_n J_n)^2}{4(1+r_n)((1+r_n)I_n- r_n J_n)},\label{I-J}
\end{align}
where the second equality follows from substituting \eqref{xi_eq}.

\medskip

{\bf Step 1.} Let $\Sigma_{N-1}^{(1)}>0$ and $\Sigma_{N-1}^{(2)}>0$ be given. From \eqref{terminal}, the values of $I_{N-1}$, $J_{N-1}$, and $\lambda_N$ are determined, and each is strictly positive. Substituting these into \eqref{fn_def} with $n=N-1$ yields
\begin{align}
f_{N-1}(\xi_{N-1})=\tfrac{\Sigma_{N-1}^{(2)}}{\Sigma_{N-1}^{(1)}+\sigma_w^2} \Big(\big(1+2\tfrac{\Sigma_{N-1}^{(1)}}{\sigma_w^2} \big) \xi_{N-1}^2 + \xi_{N-1} -2 \Big)=0. \label{f_N-1}
\end{align} 
Since 
\begin{align*}
f_{N-1}(0)=-2<0 \quad \textrm{and} \quad f_{N-1}\Big(\sqrt{\tfrac{\sigma_w^2}{\Sigma_{N-1}^{(1)}}}\Big)=\sqrt{\tfrac{\sigma_w^2}{\Sigma_{N-1}^{(1)}}}\Big(1+\sqrt{\tfrac{\sigma_w^2}{\Sigma_{N-1}^{(1)}}}\Big)>0,
\end{align*}
and $f_{N-1}$ is quadratic, there exists a unique root $\xi_{N-1}$ of \eqref{f_N-1} satisfying
\begin{align}
0<\xi_{N-1}<\sqrt{\tfrac{\sigma_w^2}{\Sigma_{N-1}^{(1)}}}.\label{xi_range}
\end{align}

Using \eqref{lr} and \eqref{IJ}, the values of $\lambda_{N-1}, r_{N-1}, \Sigma_{N-2}^{(1)}, \Sigma_{N-2}^{(2)}, I_{N-2}$, and $J_{n-2}$ are determined, and each is strictly positive by \eqref{xi_range}. 
We verify the second-order condition \eqref{SOC} for $n=N-1$:
\begin{align*}
(1+r_{N-1})((1+r_{N-1})I_{N-1}-r_{N-1}J_{N-1})=(1+r_{N-1})I_{N-1}>0,
\end{align*}
where the equality uses $I_{N-1}=J_{N-1}$. The second-order condition and \eqref{I-J} imply $I_{N-2}\leq J_{N-2}$.

In summary, the following inequalities hold for $n=N-2$:
\begin{align}
0<\xi_{n+1}< \sqrt{\tfrac{\sigma_w^2}{\Sigma_{n+1}^{(1)}}}, \quad \lambda_{n+1}>0, \quad r_{n+1}>0, \quad \Sigma_{n}^{(1)}>0,\quad \Sigma_{n}^{(2)}>0, \quad J_{n}\geq I_{n}>0. \label{ineqs}
\end{align}

{\bf Step 2.}
As the induction hypothesis, suppose \eqref{ineqs} holds for a given $n$. We determine $\xi_n, \lambda_n, r_n, \Sigma_{n-1}^{(1)},\Sigma_{n-1}^{(2)},I_{n-1}$, and $J_{n-1}$, and verify that they satisfy \eqref{ineqs} for $n-1$.

From \eqref{Sigma_eq}, \eqref{xi_def}, \eqref{lr}, and \eqref{IJ}, we obtain that
\begin{align}
\tfrac{\Sigma_n^{(2)}}{\sigma_w^2}-\tfrac{\Sigma_n^{(1)}}{\sigma_w^2}J_n =\tfrac{(1+\xi_{n+1})\Sigma_n^{(2)}}{(\Sigma_{n}^{(1)}+\sigma_w^2) \xi_{n+1} + \sigma_w^2}. \label{J_exp}
\end{align}
Using \eqref{ineqs} this yields
\begin{align}
\begin{split}\label{fn_sign}
f_n(0)&=-2I_n<0, \quad
f_n \Big(\sqrt{\tfrac{\sigma_w^2}{\Sigma_{n}^{(1)}}}\Big) = \sqrt{\tfrac{\sigma_w^2}{\Sigma_{n}^{(1)}}}\Big(1+\sqrt{\tfrac{\sigma_w^2}{\Sigma_{n}^{(1)}}}\Big)\tfrac{(1+\xi_{n+1})\Sigma_n^{(2)}}{(\Sigma_{n}^{(1)}+\sigma_w^2) \xi_{n+1} + \sigma_w^2}>0.
\end{split}
\end{align}
Since $f_n$ is a cubic polynomial with nonpositive leading coefficient, \eqref{fn_sign} implies a unique root $\xi_n$ of \eqref{fn_def} satisfying 
\begin{align}
0<\xi_{n}<\sqrt{\tfrac{\sigma_w^2}{\Sigma_{n}^{(1)}}}.\label{xin_range}
\end{align}

Given $\xi_n$, the quantities $\lambda_{n}, r_{n}, \Sigma_{n-1}^{(1)}, \Sigma_{n-1}^{(2)}, I_{n-1}$, and $J_{n-1}$ follow from \eqref{lr} and \eqref{IJ} and are strictly positive by \eqref{xin_range}. The second-order condition \eqref{SOC} reads
\begin{align}
(1+r_n)I_n - r_n J_n &= I_n+ \tfrac{\Sigma_n^{(1)}}{\sigma_w^2}(I_n-J_n) \xi_n 
=\tfrac{\sigma_w^2 \xi_n(1+\xi_n)(1+\xi_{n+1})\Sigma_n^{(2)}}{ 2(\sigma_w^2- \Sigma_n^{(1)} \xi_n^2)((\Sigma_{n}^{(1)}+\sigma_w^2) \xi_{n+1} + \sigma_w^2)}>0, \label{SOC2}
\end{align}
where the first equality follows from \eqref{lr}, the second from \eqref{fn_def} and \eqref{J_exp}, and the inequality from \eqref{ineqs} and \eqref{xin_range}. Together with \eqref{I-J}, this yields $J_{n-1}\geq I_{n-1}$. 

Thus, starting from $(\xi_{n+1}, \lambda_{n+1}, r_{n+1}, \Sigma_n^{(1)}, \Sigma_n^{(2)}, J_n, I_n)$ satisfying \eqref{ineqs}, we obtain \\$(\xi_{n}, \lambda_{n}, r_{n}, \Sigma_{n-1}^{(1)}, \Sigma_{n-1}^{(2)}, J_{n-1}, I_{n-1})$ that also satisfy \eqref{ineqs} with $n$ replaced by $n-1$. 

\smallskip

By induction, for any given $\Sigma_{N-1}^{(1)}>0$ and $\Sigma_{N-1}^{(2)}>0$, this procedure recursively constructs $\{\xi_n,r_n,I_n,J_n\}_{n=1}^{N-1}$, $\{\lambda_n\}_{n=1}^N$, and $\{\Sigma_n^{(1)},\Sigma_n^{(2)}\}_{n=0}^{N-1}$ satisfying \eqref{lr}, \eqref{xi_eq}, \eqref{SOC}, and \eqref{IJ} for $1\leq n \leq N-1$, along with \eqref{terminal}.

\medskip

{\bf Step 3.} It remains to show that there exist $\Sigma_{N-1}^{(1)}>0$ and $\Sigma_{N-1}^{(2)}>0$ such that the condition \eqref{Sigma_0} is also satisfied. From the previous step, $\Sigma_0^{(1)}$ and $\Sigma_0^{(2)}$ are determined recursively by the terminal values $\Sigma_{N-1}^{(1)}$ and $\Sigma_{N-1}^{(2)}$. We denote this dependency by functions $\Phi,\Psi: (0,\infty)^2 \to (0,\infty)$: 
\begin{align}
\Sigma_0^{(1)}= \Phi(\Sigma_{N-1}^{(1)},\Sigma_{N-1}^{(2)}),\quad \Sigma_0^{(2)}= \Psi(\Sigma_{N-1}^{(1)},\Sigma_{N-1}^{(2)}). \label{PhiPsi}
\end{align}
To complete the proof, we must show that there exist $\ha>0$ and $\hb>0$ such that 
\begin{align}
\Phi(\ha,\hb)=\sigma_a^2, \quad \Psi(\ha,\hb)=\rho \, \sigma_a \sigma_v. \label{match}
\end{align}

In Step 2, the solution $\xi_n$ to the cubic polynomial $f_n$ is a simple real root, hence it depends continuously on the coefficients of $f_n$.\footnote{It is a well-known result that the roots of a polynomial depend continuously on its coefficients. See, for example, Theorem 3.5 in the note by Alen Alexanderian, available at \url{https://aalexan3.math.ncsu.edu/articles/polyroots.pdf}.}
From \eqref{lr} and \eqref{IJ}, it follows that $ \lambda_n, r_n, \Sigma_{n-1}^{(1)},\Sigma_{n-1}^{(2)},I_{n-1}$, and $J_{n-1}$ also depend continuously on the previous values $\Sigma_n^{(1)}, \Sigma_n^{(2)},J_n$, and $I_n$. Therefore, $\Phi$ and $\Psi$ are continuous.

Let $c>0$ be a constant. If the sequence of the coefficients $\{\xi_n, \lambda_n, r_n, \Sigma_{n-1}^{(1)}, \Sigma_{n-1}^{(2)}, I_n, J_n\}_n$ satisfies \eqref{lr}, \eqref{xi_eq}, \eqref{IJ}, and \eqref{terminal}, then $\{\xi_n, c\lambda_n, r_n, \Sigma_{n-1}^{(1)}, c\Sigma_{n-1}^{(2)},c I_n, c J_n\}_n$ also satisfies the same equations. Therefore,
\begin{align}
\Phi(a,c\, b) = \Phi(a,b), \quad \Psi(a,c\, b) = c \, \Psi(a,b) \quad \textrm{for all $a,b,c>0$.} \label{scaling}
\end{align}

Since $\xi_n>0$, \eqref{lr} implies that $\Sigma_{n-1}^{(1)}>\Sigma_n^{(1)}$, hence $\Phi(a,b)>a$. Consequently,
\begin{align}
\lim_{a\to \infty} \Phi(a,b)=\infty. \label{Phi_infty}
\end{align}
Also, \eqref{lr}, \eqref{IJ}, and \eqref{fn_def} imply that 
\begin{align}
\lim_{a \downarrow 0} \Phi(a,b)=0. \label{Phi_0}
\end{align}
By continuity of $\Phi$, \eqref{Phi_infty} and \eqref{Phi_0} ensure the existence of $\ha>0$ such that $\Phi(\ha,1)=\sigma_a^2$. Setting $\hb = \tfrac{\rho\, \sigma_a \sigma_v}{\Psi(\ha,1)}$, then \eqref{scaling} yields \eqref{match}.
\end{proof}


\begin{thebibliography}{}


\bibitem[Back(1992)]{Back} Back, K., 1992. Insider trading in continuous time. Rev. Finan. Stud. 5 (3), 387–409.

\bibitem[Back and Baruch(2004)]{BackBaruch} Back, K., Baruch, S., 2004. Information in Securities Markets: Kyle Meets Glosten and Milgrom. Econometrica 72 (2), 433–465.

\bibitem[Barardehi and Bernhardt(2025)]{Barardehi} Barardehi, Y., Bernhardt, D., 2025. Revisiting the U-shaped patterns in volatility and price impacts: Novel results using trade-time estimates. J. Financ. Mark. 74, 100971.

\bibitem[Caldentey and Stacchetti(2010)]{Caldentey} Caldentey, R., Stacchetti, E., 2010. Insider trading with a random deadline. Econometrica 78 (1), 245–283.

\bibitem[Cetin(2018)]{Cetin} Çetin, U., 2018. Financial equilibrium with asymmetric information and random horizon. Finance Stoch. 22, 97–126.

\bibitem[Cetin Danilova(2016)]{CetinDanilova} Çetin, U., Danilova, A., 2016. Markovian Nash equilibrium in financial markets with asymmetric information and related forward–backward systems. Ann. Appl. Probab. 26 (4), 1996–2029.


\bibitem[Choi et al.(2023)]{Choi} Choi, J.H., Kwon, H., Larsen, K., 2023. Trading constraints in continuous-time Kyle models. SIAM J. Control Optim. 61 (3), 1494–1512.

\bibitem[Choi et al.(2019)]{Choi_JFE} Choi, J.H., Larsen, K., Seppi, D., 2019. Information and trading targets in a dynamic market equilibrium. J. Financ. Econ. 132 (3), 22–49.

\bibitem[Collin-Dufresne and Fos(2016)]{Collin} Collin-Dufresne, P., Fos, V., 2016. Insider Trading, Stochastic Liquidity, and Equilibrium Prices. Econometrica 84 (4), 1441–1475.

\bibitem[Hasbrouck(1991)]{Hasbrouck} Hasbrouck, J., 1991. Measuring the Information Content of Stock Trades. J. Financ. 46 (1), 179–207.

\bibitem[Jain and Joh(1988)]{Jain} Jain, P.C., Joh, G.H., 1988. The dependence between hourly prices and trading volume. J. Financ. Quant. Anal. 23 (3), 269–283.

\bibitem[Kyle(1985)]{Kyle} Kyle, A.S., 1985. Continuous auctions and insider trading. Econometrica 53 (6), 1315–1335.
\end{thebibliography}
\end{document}